\documentclass[aps,prd,preprint,notitlepage,nofootinbib]{revtex4-2}

\usepackage{amsmath,amssymb,amsthm,mathtools,bm}
\usepackage{hyperref}
\usepackage{physics}
\usepackage{graphicx}
\usepackage{enumitem}
\usepackage{placeins}
\usepackage[T1]{fontenc}
\usepackage{lmodern}

\usepackage{xcolor}
\definecolor{linkblack}{RGB}{20,20,20}
\definecolor{citeblue}{RGB}{0,70,140}
\definecolor{urlblue}{RGB}{0,90,120}
\usepackage{float}
\hypersetup{
  colorlinks=true,
  linkcolor=linkblack,
  citecolor=citeblue,
  urlcolor=urlblue
}

\newtheorem{proposition}{Proposition}

\begin{document}

\title{Anomaly-driven evaporation endpoints of a two-dimensional regular black hole}

\author{Damien A. Easson}
\email{easson@asu.edu}
\affiliation{Department of Physics, Arizona State University, Tempe, Arizona 85287, USA}
\affiliation{Beyond Center for Fundamental Concepts in Science, Arizona State University, Tempe, Arizona 85287, USA}

\begin{abstract}
Spherical reduction of four-dimensional minimally coupled matter yields a
two-dimensional theory with dilaton-coupled matter rather than minimally
coupled conformal matter. We use this distinction to revisit the backreacted late-time endpoint problem
for the regular two-dimensional Bardeen-like black hole considered by Barenboim, Frolov,
and Kunstatter. Replacing the Polyakov quantum sector by
the dilaton-coupled anomaly model of Fabbri, Farese, and Navarro-Salas (FFN), we derive the corresponding semiclassical field equations and classify the
asymptotically allowed late branches at finite radius. For any quiescent finite-radius branch with finite nonzero conformal factor, the late-time
mixed equation enforces \(J'(r_\infty)=0\), and hence
\(r_\infty=\sqrt{2}\,\ell\), independently of the local dilaton-anomaly
convention. For finite-radius null branches satisfying the stated state-tail
assumptions, the ordinary strong-cosmic-censorship-restoring exponential null
boundary is excluded. Generic power-law branches \(e^{2\rho}\sim v^{-p}\) with \(p>1\) are likewise
excluded, except for the borderline case \(p=2\), which is the only remaining
null loophole of this type. In the FFN model, the settled realization of this
loophole carries finite affine flux and requires the stronger state-tail decay
\(s_\phi=O(v^{-2})\).  
The natural finite-radius outcome is
remnant-like, while the surviving null branch is a highly constrained soft-null
alternative.
\end{abstract}

\maketitle

\newpage

\section{Introduction}

Regular black holes provide a natural laboratory for studying nonsingular
evaporation endpoints~\cite{Bardeen:1968,AyonBeato:2000zs,Hayward:2005gi},
while two-dimensional semiclassical models offer a controlled setting for
analyzing backreaction explicitly~\cite{Callan:1992rs,Russo:1992ax}. A recent
example is the Bardeen-like model of Barenboim, Frolov, and Kunstatter
(BFK)~\cite{Barenboim:2025bfk}, where the semiclassical dynamics were studied
numerically using the standard two-dimensional conformal anomaly, equivalently
the Polyakov description for minimally coupled two-dimensional matter
~\cite{Polyakov:1981rd}. Their simulations exhibit trapped and anti-trapped
regions and, for microscopic black holes, end in a spacetime free of apparent
and Cauchy horizons.

This Polyakov-based result is highly suggestive, but spherical reduction of a
minimally coupled four-dimensional scalar gives a nonminimally coupled
two-dimensional matter model with explicit dilaton coupling, not the minimally
coupled conformal matter described by the Polyakov sector
\cite{Mukhanov:1994ax,Kummer:1997ys,Grumiller:2002nm}. This is significant because the
dilaton-coupled quantum stress tensor is not fixed by the trace anomaly alone:
its state-dependent functions obey coupled nonchiral equations. Related
four-dimensional anomaly-induced studies likewise emphasize the importance of
stress-tensor structure and state selection in dynamical Hawking emission and
backreaction~\cite{Lowe:2025tik,Lowe:2026kvi}.

The present analysis also fits within a broader anomaly-driven
two-dimensional endpoint program, where reduced Polyakov models have led to
analytic nonequilibrium endpoint dynamics in multi-horizon black-hole
evaporation~\cite{Easson:2025ekn,Easson:2026kng}. Here we ask whether the BFK endpoint picture is stable once the quantum model is
replaced by the dilaton-coupled anomaly structure suggested by spherical
reduction.
We therefore revisit the endpoint problem for the Bardeen-like geometry using
the concrete anomaly-driven stress tensor of Fabbri, Farese, and Navarro-Salas
(FFN)~\cite{Fabbri:2003vy}. This model retains explicit dilaton-dependent stress-tensor contributions and
the associated nonchiral state equations absent from the minimally coupled
Polyakov description.

We proceed analytically, focusing on the asymptotic branches selected by the
dilaton-coupled system. We show that quiescent finite-radius branches obey the selector
\(J'(r_\infty)=0\), giving \(r_\infty=\sqrt2\,\ell\), where \(\ell\) is the
Bardeen core scale. We further show that the familiar SCC-restoring null blow-up and generic
power-law null branches \(e^{2\rho}\sim v^{-p} \), \(p\neq2 \), are excluded
within the stated asymptotic assumptions. The borderline $p=2$ case
survives only as a constrained soft-null loophole. Thus the endpoint problem is
reclassified into a benign finite-radius remnant-like branch and a softer null
alternative. While global endpoint selection remains a time-dependent backreaction and
state-selection problem, the local asymptotics already show that the Polyakov
picture is not robust once the dilaton-coupled anomaly is used.

\section{The Bardeen model}

We work with the metric in double-null gauge,
\begin{equation}
ds^2=-e^{2\rho(u,v)}\,du\,dv .
\end{equation}
The Bardeen-like classical model is specified by the function \(J(r)\), where
\(r(u,v)\) is the areal-radius/dilaton field:
\begin{equation}
J(r)=\frac{(r^2+\ell^2)^{3/2}}{r^2}.
\end{equation}
In the BFK frame~\cite{Barenboim:2025bfk}, the classical action may be written as
\begin{equation}
S_G=\frac{1}{2G}\int d^2x\,\sqrt{-g}\,
\Big[\Phi(r)R+\Phi''(r)(\nabla r)^2+\Phi''(r)\Big],
\qquad
J(r)=\Phi'(r).
\end{equation}
Thus, \(\Phi(r)\) is an antiderivative of the model function \(J(r)\), defined
up to an irrelevant constant.

The resulting field equation is
\begin{equation}
-J\big(\nabla_\mu\nabla_\nu r-g_{\mu\nu}\nabla^2 r\big)
+\frac{J'}{2}g_{\mu\nu}(\nabla r)^2
-\frac{J'}{2}g_{\mu\nu}
=G\,T_{\mu\nu}.
\end{equation}
We use the standard matter stress-tensor convention
\begin{equation}
\delta S_m=-\frac12\int d^2x\,\sqrt{-g}\,
T_{\mu\nu}\,\delta g^{\mu\nu}.
\end{equation}
In the semiclassical equations below, \(T_{\mu\nu}\) denotes the renormalized
expectation value \(\langle T_{\mu\nu}\rangle\) in the chosen quantum state.
The state dependence is encoded in the functions \(t_u\), \(t_v\), and
\(s_\phi\).

In our gauge,
\begin{equation}
g_{uv}=-\frac12 e^{2\rho},\qquad g_{uu}=g_{vv}=0,
\end{equation}
and the useful geometric identities are
\begin{equation}
\begin{aligned}
\nabla_u\nabla_u r &= r_{uu}-2\rho_u r_u,
&\qquad
\nabla_v\nabla_v r &= r_{vv}-2\rho_v r_v,\\
\nabla_u\nabla_v r &= r_{uv},
&\qquad
(\nabla r)^2 &= -4e^{-2\rho}r_u r_v,\\
\nabla^2 r &= -4e^{-2\rho}r_{uv}.
\end{aligned}
\end{equation}
The mixed component and the $vv$ null component give
\begin{equation}
J\,r_{uv}+J' r_u r_v+\frac14 e^{2\rho}J' = G\,T_{uv},
\label{eq:mixed-exact}
\end{equation}
and
\begin{equation}
-J\,(r_{vv}-2\rho_v r_v)=G\,T_{vv},
\label{eq:vv-exact}
\end{equation}
with the $uu$ equation obtained by $u\leftrightarrow v$.

All endpoint statements below are obtained by inserting late-time asymptotic ans\"atze into these exact equations and comparing powers of $v$.

\section{Dilaton-coupled anomaly}

We model the quantum contribution using the dilaton-coupled anomaly model of
FFN~\cite{Fabbri:2003vy}. In our spherically reduced normalization,
the dilaton \(\phi\) is defined by
\begin{equation}
e^{-2\phi}=r^2,
\qquad
\phi=-\ln r.
\label{eq:phidef}
\end{equation}

For orientation, the minimally coupled Polyakov sector is generated by the standard
nonlocal effective action~\cite{Polyakov:1981rd,Birrell:1982ix},
\begin{equation}
S_{\rm P}=-\frac{N}{96\pi}\int d^2x\,\sqrt{-g}\;R\,\Box^{-1}R.
\end{equation}
The associated two-dimensional trace-anomaly stress-tensor construction underlies
the standard Hawking-flux analysis~\cite{Christensen:1977jc}.
By contrast, spherical reduction of a four-dimensional minimally coupled scalar
\begin{equation}
S_m=-\frac{1}{8 \pi}\int d^4x\,\sqrt{-g}\,(\nabla f)^2,
\end{equation}
with metric $ds_{(4)}^2 =ds_{(2)}^2  + B(r) d\Omega^2$,
where $B(r)=e^{-2\phi}=r^2$ leads to a dilaton-coupled two-dimensional matter action of the form
\begin{equation}
S_m=-\frac12\int d^2x\,\sqrt{-g}\,B(r)(\nabla f)^2.
\end{equation}
Thus, the anomaly acquires an explicit $\phi$-dependence, involving terms built from $R$, $(\nabla\phi)^2$, and $\Box\phi$. The precise dilaton-dependence and associated one-loop effective action are
known to be subtle: the literature contains both historical disagreements over
local anomaly terms, and more recent constructions emphasizing Weyl-invariant
ambiguities and state-dependent stress tensors
\cite{Bousso:1997wi,Kummer:1997ys,Kummer:1998sp,Nojiri:1997ab,
Nojiri:1997gf,Fabbri:2003vy,Wu:2023uyb}.
For these reasons, we work directly with the FFN stress tensor and its
associated state equations~\cite{Fabbri:2003vy}, rather than assuming a unique
induced anomaly action.

In conformal gauge, with \(t_v(u,v)\) denoting the nonchiral
state-dependent function associated with the \(vv\) component, the relevant
stress-tensor components are
\begin{equation}
\langle T_{vv}\rangle
=-\frac{1}{12\pi}t_v
-\frac{1}{12\pi}(\rho_v^2-\rho_{vv})
+\frac{1}{2\pi}(\rho_v\phi_v+\rho\,\phi_v^2),
\end{equation}
\begin{equation}
\langle T_{uv}\rangle
=-\frac{1}{12\pi}(\rho_{uv}+3\phi_u\phi_v-3\phi_{uv}).
\end{equation}

The stress tensor is in single-field normalization. For \(N\)
identical dilaton-coupled scalar fields, the stress tensor is multiplied
by \(N\); equivalently, after insertion into the semiclassical field equations,
one replaces \(G\) by \(GN\) in the anomaly-induced terms.

Using \eqref{eq:phidef} we have
\begin{equation}
3\phi_u\phi_v-3\phi_{uv}=3\frac{r_{uv}}{r},
\end{equation}
so that
\begin{equation}
\langle T_{uv}\rangle
=-\frac{1}{12\pi}\left(\rho_{uv}+3\frac{r_{uv}}{r}\right).
\end{equation}
The exact mixed equation becomes
\begin{equation}
\left(J+\frac{G}{4\pi r}\right)r_{uv}
+\frac{G}{12\pi}\rho_{uv}
+J' r_u r_v
+\frac14 e^{2\rho}J'=0.
\label{eq:mixed-ffns}
\end{equation}
Likewise,
\begin{equation}
J(r_{vv}-2\rho_v r_v)
=
\frac{G}{12\pi}t_v
+\frac{G}{12\pi}(\rho_v^2-\rho_{vv})
+\frac{G}{2\pi}\frac{\rho_v r_v}{r}
-\frac{G}{2\pi}\rho\,\frac{r_v^2}{r^2}.
\label{eq:vv-ffns}
\end{equation}
Equation \eqref{eq:vv-ffns} is simply the $vv$ constraint \eqref{eq:vv-exact} rewritten after substitution of the stress tensor.
The exact state equations are
\begin{equation}
\partial_u t_v + 3\,\partial_v\!\left(\frac{r_{uv}}{r}\right)-6\frac{r_v}{r}\,s_\phi=0,
\label{eq:state-v}
\end{equation}
\begin{equation}
\partial_v t_u + 3\,\partial_u\!\left(\frac{r_{uv}}{r}\right)-6\frac{r_u}{r}\,s_\phi=0,
\label{eq:state-u}
\end{equation}
where \(s_\phi\) is a rescaled version of the state-dependent dilaton
source, defined by
\begin{equation}
-\frac{1}{2\pi}s_\phi
\equiv
\left\langle \frac{\delta S}{\delta\phi}\right\rangle_{\rho=0}.
\end{equation}

\section{Late-time selector and asymptotic endpoints}
\label{sec:selector}

We now derive the local late-time endpoint classification: an
ambiguity-independent selector for quiescent finite-radius branches, a
conditional exclusion of the ordinary SCC-restoring null branch, and the special
\(p=2\) null-boundary loophole.

\subsection{Late-time selector}

To allow for the historical ambiguity in the local dilaton-coupled anomaly~\cite{Bousso:1997wi,Kummer:1997ys,Nojiri:1997ab,Nojiri:1997gf,Fabbri:2003vy}, it is convenient to work with the one-parameter family
\begin{equation}
\label{eq:trace-family}
\langle T^\mu{}_{\mu}\rangle_\alpha
=
c\Bigl(R-6(\nabla\phi)^2+\alpha\,\Box\phi\Bigr),
\end{equation}
where \(c>0\) is the overall anomaly coefficient and \(\alpha\) is a constant
ambiguity parameter. In terms of $r$, this becomes
\begin{equation}
\label{eq:trace-family-r}
\langle T^\mu{}_{\mu}\rangle_\alpha
=
c\left(
R-\alpha\frac{\Box r}{r}
+(\alpha-6)\frac{(\nabla r)^2}{r^2}
\right).
\end{equation}
With our convention $R=8e^{-2\rho}\rho_{uv}$, the corresponding local contribution to the mixed component is
\begin{equation}
\label{eq:Tuv-alpha}
G\,T_{uv}^{(\alpha)}
=
-\gamma\,\rho_{uv}
-\frac{\alpha\gamma}{2}\frac{r_{uv}}{r}
+\frac{(\alpha-6)\gamma}{2}\frac{r_u r_v}{r^2},
\end{equation}
where $\gamma \equiv 2Gc >0$ absorbs the overall normalization. Combining \eqref{eq:Tuv-alpha} with \eqref{eq:mixed-exact} gives
\begin{equation}
\label{eq:uv-alpha}
\left(J+\frac{\alpha\gamma}{2r}\right)r_{uv}
+
\left(
J'-\frac{(\alpha-6)\gamma}{2r^2}
\right)r_u r_v
+\gamma\,\rho_{uv}
+\frac14 e^{2\rho}J'
=0.
\end{equation}

\begin{proposition}[late-time selector]
\label{prop:selector}
Assume that the anomaly-driven model admits a late outgoing branch such that, as $v\to\infty$,
\begin{equation}
\label{eq:quiescent-assumptions}
r(u,v)\to r_\infty(u)>0,
\qquad
r_v\to0,
\qquad
r_{uv}\to0,
\qquad
\rho_{uv}\to0,
\qquad
r_u r_v\to0,
\end{equation}
and
\begin{equation}
\label{eq:finite-e2rho}
e^{2\rho(u,v)}\to E(u),
\qquad
0<E(u)<\infty.
\end{equation}
Then
\begin{equation}
\label{eq:selector}
J'(r_\infty)=0.
\end{equation}
For the Bardeen model~\footnote{The same selector mechanism applies to any regular model written in
the BFK form with a smooth model function \(J(r)\). For example, for a
Hayward-type model~\cite{Hayward:2005gi} with
\(J_H(r)=(r^3+\ell^3)/r^2\), one finds
\(J_H'(r)=1-2\ell^3/r^3\), so \(J_H'(r_\infty)=0\) selects
\(r_\infty=2^{1/3}\ell\). When \(J(r)\) has a unique positive minimum, the
selector picks that radius; in the associated static family this is the
extremal radius.},
\begin{equation}
\label{eq:bardeen-selector}
r_\infty=\sqrt{2}\,\ell.
\end{equation}
\end{proposition}

\begin{proof}
Under the assumptions \eqref{eq:quiescent-assumptions}--\eqref{eq:finite-e2rho}, the derivative terms in \eqref{eq:uv-alpha} vanish as $v\to\infty$, and the equation reduces to
\begin{equation}
\frac14 E(u)\,J'(r_\infty)=0.
\end{equation}
Since $E(u)>0$, this implies \eqref{eq:selector}. For Bardeen,
\begin{equation}
J'(r)=\frac{(r^2-2\ell^2)\sqrt{r^2+\ell^2}}{r^3}.
\end{equation}
The unique positive zero is therefore \(r=\sqrt{2}\ell\), which proves
\eqref{eq:bardeen-selector}.

\end{proof}

The key point is that the selector \eqref{eq:selector} is independent of the
ambiguity parameter \(\alpha\) in \eqref{eq:trace-family} and, therefore, the picked-out radius
is not merely an artifact of one preferred local anomaly convention;
within the BFK form, its value is fixed by the stationary points of \(J(r)\). On a quiescent finite-radius branch
with finite nonzero \(e^{2\rho}\), the derivative and flux terms in the mixed
equation become asymptotically subleading, so the remaining balance is governed
by the classical potential term \(\frac14 e^{2\rho}J'\). Within this asymptotic
class, the endpoint radius must therefore lie at a stationary point of the
original Bardeen function \(J(r)\). The dilaton anomaly enters through the
allowed approach to this radius and it changes which asymptotic branches are
consistent and fixes the correlated state tails needed to support the surviving
null alternatives. Null endpoints with \(e^{2\rho}\to0\) are
classified separately below.

\subsection{Conditional exclusion of the SCC-restoring null branch}

We now show that the usual nonextremal SCC-restoring null branch is
asymptotically incompatible with finite-radius settling, unless the ingoing
state sector supplies a competing leading tail. Using the FFN stress tensor
in the \(vv\) constraint \eqref{eq:vv-exact}, the explicit dilaton
contributions near a settling finite-radius branch are suppressed by factors
of \(r_v\), while the state term \(t_v\) must be independently controlled. Under
our asymptotic state assumptions below, the leading behavior is of
Polyakov type,
\begin{equation}
\label{eq:Tvv-leading}
T_{vv}\sim -\frac{1}{12\pi}\bigl(\rho_v^2-\rho_{vv}\bigr).
\end{equation}

\begin{proposition}[conditional exclusion of the ordinary SCC branch]
\label{prop:scc-exclusion}
Consider a late finite-radius branch satisfying
\begin{equation}
r(u,v)\to r_\infty(u)>0,
\qquad
r_v\to0,
\qquad
r_{vv}\to0
\qquad (v\to\infty).
\end{equation}
Suppose the null-boundary asymptotics below are differentiable, so that their
leading forms may be differentiated term by term, and suppose the explicit
dilaton terms in \eqref{eq:vv-ffns} are subleading relative to the
Polyakov-type term. More explicitly, in the power-law case we require
\begin{equation}
\frac{\rho_v r_v}{r}=o(v^{-2}),
\qquad
\rho\,\frac{r_v^2}{r^2}=o(v^{-2}),
\end{equation}
while in the exponential case the corresponding explicit dilaton terms are
\(o(1)\).

Then:
\begin{enumerate}
\item If
\begin{equation}
e^{2\rho}\sim e^{-\beta v},
\qquad \beta>0,
\end{equation}
and the ingoing state tail satisfies \(t_v=o(1)\), this branch is excluded.

\item If
\begin{equation}
e^{2\rho}\sim v^{-p},
\qquad p>1,
\end{equation}
and the ingoing state tail satisfies \(t_v=o(v^{-2})\), this branch is excluded
for all \(p\neq2\).
\end{enumerate}
Within this finite-radius exponential/power-law class, the only surviving
null-boundary loophole is the special \(p=2\) branch,
\begin{equation}
\label{eq:p2-loophole}
e^{2\rho}\sim \frac{A(u)}{v^2}.
\end{equation}
\end{proposition}

\begin{proof}
For the exponential case,
\begin{equation}
e^{2\rho}\sim e^{-\beta v}
\quad\Longrightarrow\quad
\rho_v\to-\frac{\beta}{2},
\qquad
\rho_{vv}\to0.
\end{equation}
By the regularity and state-tail hypotheses, the explicit dilaton terms and
state contribution are both \(o(1)\). Hence
\begin{equation}
T_{vv}\to -\frac{\beta^2}{48\pi}\neq0.
\end{equation}
But \(J(r_\infty)\) is finite and \(\rho_v\) is bounded, so the left-hand side
of \eqref{eq:vv-exact} tends to zero when \(r\to r_\infty\) with
\(r_v,r_{vv}\to0\), giving a contradiction. Hence the exponential branch is
excluded.

For the power-law case,
\begin{equation}
e^{2\rho}\sim v^{-p}
\quad\Longrightarrow\quad
\rho_v=-\frac{p}{2v},
\qquad
\rho_{vv}=\frac{p}{2v^2}.
\end{equation}
By the regularity hypothesis, the explicit dilaton terms in \eqref{eq:vv-ffns}
are \(o(v^{-2})\), and the hypothesis \(t_v=o(v^{-2})\) removes any competing
state tail at the same order. Therefore
\begin{equation}
\label{eq:Tvv-power}
T_{vv}\sim \frac{p(2-p)}{48\pi\,v^2}.
\end{equation}
For \(p\neq2\), inserting \eqref{eq:Tvv-power} into \eqref{eq:vv-exact} yields
the asymptotic balance
\begin{equation}
-J(r_\infty)\,r_{vv}
-\frac{pJ(r_\infty)}{v}\,r_v
\sim \frac{A_p(u)}{v^2},
\end{equation}
with \(A_p(u)\neq0\) for \(p\neq2\). Solving this asymptotically gives
\begin{equation}
r_v\sim \frac{c(u)}{v},
\qquad
c(u)\neq0,
\end{equation}
and hence
\begin{equation}
r(u,v)\sim r_\infty(u)+c(u)\ln v,
\end{equation}
contradicting the assumed finite limit \(r\to r_\infty\). Thus all
\(p>1\) with \(p\neq2\) are excluded.
\end{proof}

The exceptional nature of the \(p=2\) case is already visible at the level of
the local Polyakov-type term: for \(e^{2\rho}\sim v^{-p}\),
\[
\rho_v^2-\rho_{vv}=\frac{p(p-2)}{4v^2}.
\]
Thus any perturbation of the power away from \(p=2\) restores the \(v^{-2}\)
term that drives the logarithmic finite-radius obstruction in
Proposition~\ref{prop:scc-exclusion}. The exceptional value is therefore not
specific to the FFN dilaton coupling; it follows from the universal
Polyakov-type null combination \(\rho_v^2-\rho_{vv}\). What \emph{is} model-dependent
is the fate of the branch after this leading cancellation: in the FFN system,
the subleading dilaton terms and state equations reduce it to the constrained
soft-null sector analyzed below. What remains open is whether the correlated
state tails required by this branch can arise from global collapse-normalized
state data.

Proposition~\ref{prop:scc-exclusion} therefore gives a conditional asymptotic
exclusion of the ordinary SCC-restoring null singularity at finite radius.
Establishing the required decay of \(t_v\) directly from global
collapse-normalized state data is a separate state-selection problem.
This exclusion is distinct from a full Cauchy-horizon stability
analysis. It removes the local finite-radius null branches with the usual
affine amplification familiar from Cauchy-horizon and mass-inflation
analyses~\cite{Poisson:1990eh,Ori:1991zz,Dafermos:2003wr}, while leaving the
global horizon structure of the dilaton-coupled anomaly model as a separate
problem. In the reduced Polyakov model, the corresponding affine flux coefficient can
be characterized directly in state space~\cite{Easson:2025uca}; here the point
is that the FFN late-time equations obstruct this local branch under the stated
state-tail conditions.
Our resulting branch classification is summarized in
Fig.~\ref{fig:endpoint-classification}.

\begin{figure}[t]
\centering
\includegraphics[width=\columnwidth]{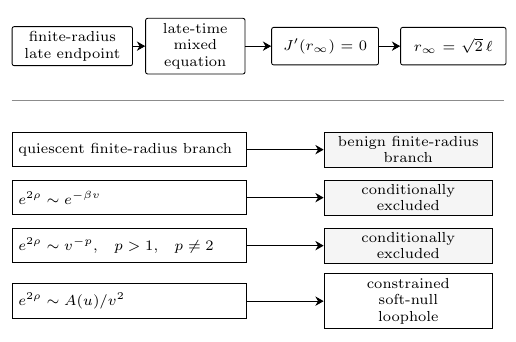}
\caption{Asymptotic endpoint classification in the dilaton-coupled
anomaly model. Finite-radius late endpoints are fixed by the late-time mixed
equation, which gives \(J'(r_\infty)=0\) and hence
\(r_\infty=\sqrt{2}\ell\) for the Bardeen model. Exponential null branches and
generic power-law branches with \(p>1\), \(p\neq2\), are conditionally excluded
under the state-tail and regularity assumptions. The remaining null
loophole is the constrained \(p=2\) soft-null branch.}
\label{fig:endpoint-classification}
\end{figure}

\subsection{The surviving $p=2$ sector}

We now return to the FFN model and impose the settled special-null ansatz
\begin{equation}
\rho=\frac12\ln A_0-\ln v+\frac{b_2(u)}{v^2}+O(v^{-3}),
\qquad
r=r_0+\frac{a_2(u)}{v^2}+O(v^{-3}),
\end{equation}
with the \(1/v\) coefficients in both fields set to zero. The leading coefficient \(A_0\) is taken constant,
defining the settled representative of the \(p=2\) family.~\footnote{A more
general \(A(u)\) can be absorbed locally by a residual reparametrization of the
outgoing null coordinate, or, in a fixed collapse-normalized frame, would
correspond to additional leading outgoing state data in the \(uu\) sector.}

We display the relevant coefficient relations explicitly, since the consistency
of this branch depends on correlations among the $vv$ constraint, the state
equations, and the $uu$ constraint.

From \eqref{eq:mixed-ffns}, the leading \(v^{-2}\) coefficient gives
\begin{equation}
J'(r_0)=0,
\qquad
r_0=\sqrt{2}\,\ell ,
\end{equation}
so the settled branch is tied to the same selected radius.
At the next order, the same mixed equation yields
\begin{equation}
\left(J_0+\frac{G}{4\pi r_0}\right)a_2'(u)+\frac{G}{12\pi}b_2'(u)=0,
\label{eq:a2b2}
\end{equation}
where $J_0:=J(r_0)$.

From \eqref{eq:vv-ffns},
\begin{equation}
t_v=\frac{\tau_4(u)}{v^4}+O(v^{-5}),
\qquad
T_{vv}=-\frac{2J_0 a_2(u)}{G\,v^4}+O(v^{-5}),
\end{equation}
so, up to the constant normalization of the affine parameter, the flux tends to the finite limit
\begin{equation}
T_{kk}=e^{-4\rho}T_{vv}\to -\frac{2J_0 a_2(u)}{G A_0^2}.
\end{equation}
The state equations then determine the asymptotic state tails required by the settled $a_2/v^2$ branch, for which
\begin{equation}
r_v=-\frac{2a_2}{v^3}+O(v^{-4}),
\qquad
r_{uv}=-\frac{2a_2'(u)}{v^3}+O(v^{-4}).
\end{equation}
Substituting into \eqref{eq:state-v} gives
\begin{equation}
\partial_u t_v+\frac{18}{r_0}\frac{a_2'(u)}{v^4}
+\frac{12}{r_0}\frac{a_2(u)s_\phi}{v^3}
+O(v^{-5})=0.
\end{equation}
Hence the dilaton-source function must decay along the endpoint. The weakest allowed decay is
\begin{equation}
s_\phi=\frac{\sigma_1(u)}{v}+O(v^{-2}),
\end{equation}
for which the leading nontrivial coefficient of \eqref{eq:state-v} becomes
\begin{equation}
\tau_4'(u)+\frac{18}{r_0}a_2'(u)+\frac{12}{r_0}a_2(u)\sigma_1(u)=0,
\label{eq:tau4prime}
\end{equation}
where
\begin{equation}
t_v=\frac{\tau_4(u)}{v^4}+O(v^{-5}).
\end{equation}

The $vv$ constraint algebraically fixes the coefficient $\tau_4$:
\begin{equation}
\tau_4
=
\frac{24\pi J_0}{G}a_2
+2b_2
-\frac{12}{r_0}a_2 .
\label{eq:tau4-algebraic}
\end{equation}
Differentiating and using \eqref{eq:a2b2} gives
\begin{equation}
\tau_4'=-\frac{18}{r_0}a_2' .
\end{equation}
Combining this with \eqref{eq:tau4prime} yields
\begin{equation}
a_2(u)\sigma_1(u)=0.
\label{eq:a2sigma-strong}
\end{equation}

The companion equation \eqref{eq:state-u} is equally informative. Since
\begin{equation}
r_u=\frac{a_2'(u)}{v^2}+O(v^{-3}),
\end{equation}
Eq.~\eqref{eq:state-u} implies
\begin{equation}
\partial_v t_u
-\frac{6}{r_0}\,
\frac{a_2''(u)+a_2'(u)\sigma_1(u)}{v^3}
+O(v^{-4})=0 .
\end{equation}
Thus the leading state tail must be
\begin{equation}
t_u=\frac{\eta_2(u)}{v^2}+O(v^{-3}),
\qquad
\eta_2(u)=-\frac{3}{r_0}\big(a_2''(u)+a_2'(u)\sigma_1(u)\big).
\label{eq:eta2}
\end{equation}

The \(uu\) constraint must then be combined with the local FFN
\(-\rho_{uu}\) term, which contributes at the same order as the state tail. At
order \(v^{-2}\) one finds
\begin{equation}
J_0 a_2''(u)=\frac{G}{12\pi}\bigl(\eta_2(u)-b_2''(u)\bigr).
\label{eq:a2uu}
\end{equation}
Differentiating \eqref{eq:a2b2} gives
\begin{equation}
b_2''(u)=-\frac{12\pi}{G}\left(J_0+\frac{G}{4\pi r_0}\right)a_2''(u).
\label{eq:b2dprime}
\end{equation}
Substituting \eqref{eq:eta2} and \eqref{eq:b2dprime} into \eqref{eq:a2uu} yields
\begin{equation}
a_2'(u)\sigma_1(u)=0.
\label{eq:a2sigma}
\end{equation}
Together, \eqref{eq:a2sigma-strong} and \eqref{eq:a2sigma} imply that any
nonzero settled \(a_2/v^2\) branch requires \(\sigma_1(u)=0\), and hence
\(s_\phi=O(v^{-2})\). If \(a_2\) vanishes identically, the affine-flux coefficient also vanishes and
the branch degenerates to the trivial settled solution. Thus a nontrivial settled \(p=2\) branch is asymptotically
consistent only with the correlated state-tail conditions
\eqref{eq:tau4-algebraic}, \eqref{eq:tau4prime}, and
\eqref{eq:eta2}.~\footnote{The corresponding \(a_1/v\) analysis is derived in
Appendix~\ref{app:a1}. It imposes stronger asymptotic conditions on the state
tails.}

\section{Collapse-normalized state data}

The one-sided collapse setup matches to a flat interior with
\begin{equation}
\rho=0,
\qquad
r=\frac{v-u}{2},
\qquad
r_u=-\frac12,
\qquad
r_v=\frac12,
\qquad
r_{uv}=0.
\end{equation}
With no incoming radiation in this initial flat region, the state equations force
\begin{equation}
s_\phi=0.
\end{equation} 

Nevertheless, the exact state equations show that nontrivial state tails are generated dynamically once the geometry becomes nontrivial. Written in integral form,
\begin{equation}
t_v(u,v)=t_v(u_0,v)-3\int_{u_0}^{u}du'\,\partial_v\!\left(\frac{r_{u'v}}{r}\right)
+6\int_{u_0}^{u}du'\,\frac{r_v}{r}s_\phi,
\end{equation}
\begin{equation}
t_u(u,v)=t_u(u,v_0)-3\int_{v_0}^{v}dv'\,\partial_u\!\left(\frac{r_{uv'}}{r}\right)
+6\int_{v_0}^{v}dv'\,\frac{r_u}{r}s_\phi.
\end{equation}
For the settled $a_2$ branch the exact asymptotic falloffs are consistent with these integral relations. There is therefore no evident asymptotic scaling mismatch: the collapse boundary conditions do not eliminate the settled $a_2/v^2$ branch at the level of these local late-time expansions.

In a collapse-normalized state regular on the initially flat region, the state
functions are fixed nonlocally by the state equations and initial data, rather
than freely tunable at late times. The hypotheses used in
Proposition~\ref{prop:scc-exclusion} are
state-selection conditions: a leading \(t_v\) tail capable of canceling the
Polyakov-type term would represent a correlated incoming-state contribution.
Deriving or excluding such a tail from global collapse data is the nonlocal
state-selection problem we leave open.

\section{Discussion}

Our analysis gives a local asymptotic reclassification of the BFK endpoint
problem once the Polyakov description is replaced by a candidate
dilaton-coupled anomaly model. The condition \(J'(r_\infty)=0\) is robust
within the one-parameter anomaly family, while
Proposition~\ref{prop:scc-exclusion} conditionally excludes the ordinary
SCC-restoring exponential branch and generic \(p\neq2\) power-law branches
under our given regularity and state-tail assumptions. Thus the detailed
no-horizon Polyakov endpoint picture is sensitive to the semiclassical matter
model.

In the FFN model, the settled \(a_2/v^2\) branch is the least tuned
null realization. When nonzero, it is tied to the same selected radius, carries
finite rather than divergent affine flux, and requires correlated state-sector
tails with \(s_\phi=O(v^{-2})\). The more rigid \(a_1/v\) branch remains logically possible at the level of
local asymptotics, but only under substantially stronger tail constraints. The serious remaining
late-time alternatives are a benign finite-radius branch and a highly
constrained soft-null branch. From this local perspective, the former is the natural candidate for the physical evaporation endpoint, while
the latter appears as a state-selected loophole whose global
realization must be tested.

Several open questions remain. Determining whether the anti-trapped or white-hole-like late phase found in
the Polyakov evolution persists requires a global analysis of the horizon
structure in the dilaton-coupled anomaly model. A more complete treatment of the nonlocal state-selection problem might both justify the required decay of $t_v$ for the excluded branches and decide between the benign and soft-null branches. It would be valuable to compare the present analytic classification with other regular two-dimensional black-hole models and with broader discussions of remnants in semiclassical dilaton gravity~\cite{Callan:1992rs,Hayward:2005gi,Cadoni:2023tse,Barenboim:2025bfk}.

The remnant-like branch naturally connects to scenarios in
which nonsingular primordial black holes, or their evaporation remnants, may significantly
contribute to the dark matter~\cite{Easson:2002tg,Davies:2024nxbh,Calza:2024fzo,Calza:2024xdh,Calza:2025mwn,Carr:2021bzv}. In the present work, however, we establish only
the local asymptotic endpoint structure in a candidate anomaly-driven model.
Stability, abundance, production history, and cosmological constraints lie outside the scope of our analysis.

Taken conservatively, our analysis shows that the Polyakov endpoint picture is
not obviously robust once the more faithful dilaton-coupled anomaly structure is
taken seriously. More provocatively, if future work demonstrates that
collapse-normalized state selection eliminates the remaining \(p=2\) soft-null
loophole, our classification would suggest the finite-radius branch as
the physical evaporation endpoint. If combined with a perturbative stability
analysis, the result would provide strong evidence for a quantum-selected
stable remnant as the endpoint of nonsingular black-hole evaporation.

\acknowledgments
This work was supported by the U.S. Department of Energy,
Office of High Energy Physics, under Award Number DE-SC0019470.

\appendix

\section{The more rigid \texorpdfstring{$a_1/v$}{a1/v} branch}
\label{app:a1}

In this appendix we present the corresponding asymptotic analysis for the more rigid
\begin{equation}
r-r_0\sim \frac{a_1(u)}{v}
\end{equation}
null branch. While not excluded outright by the local asymptotic equations, it is highly constrained compared with the settled $a_2/v^2$ branch.

We take
\begin{equation}
e^{2\rho}\sim \frac{A(u)}{v^2},
\qquad
\rho=\frac12\ln A(u)-\ln v+\frac{b_1(u)}{v}+O(v^{-2}),
\end{equation}
and
\begin{equation}
r=r_0+\frac{a_1(u)}{v}+O(v^{-2}).
\end{equation}
Then
\begin{equation}
r_v=-\frac{a_1(u)}{v^2}+O(v^{-3}),
\qquad
r_{uv}=-\frac{a_1'(u)}{v^2}+O(v^{-3}),
\qquad
\rho_{uv}=-\frac{b_1'(u)}{v^2}+O(v^{-3}).
\end{equation}

Substituting into \eqref{eq:mixed-ffns} gives at order $v^{-2}$
\begin{equation}
-\left(J_0+\frac{G}{4\pi r_0}\right)a_1'(u)
-\frac{G}{12\pi}b_1'(u)
+\frac14 A(u)J'(r_0)=0.
\label{eq:a1-mixed}
\end{equation}
Unlike the settled \(a_2\) branch, the \(a_1\) branch does not by itself impose
\(J'(r_0)=0\). That conclusion would require additional settling conditions on
the \(1/v\) coefficients.

The exact $vv$ constraint \eqref{eq:vv-ffns} is more restrictive. Since
\begin{equation}
r_{vv}-2\rho_v r_v = O(v^{-4}),
\end{equation}
the leading $v^{-3}$ coefficient must vanish. The local dilaton term contributes
\begin{equation}
\frac{G}{2\pi}\frac{\rho_v r_v}{r}
=
\frac{G}{2\pi}\frac{a_1(u)}{r_0\,v^3}
+O(v^{-4}),
\end{equation}
so consistency requires
\begin{equation}
t_v=\frac{\tau_3(u)}{v^3}+O(v^{-4}),
\qquad
\tau_3(u)=-\frac{6a_1(u)}{r_0}.
\label{eq:a1-tau3}
\end{equation}

The first state equation \eqref{eq:state-v} now gives
\begin{equation}
\partial_u t_v
+\frac{6a_1(u)}{r_0}\frac{s_\phi}{v^2}
+\frac{6a_1'(u)}{r_0\,v^3}
+O(v^{-4})=0.
\end{equation}
Thus the dilaton-source function cannot approach a nonzero constant if $a_1\neq0$. Write instead
\begin{equation}
s_\phi=\frac{\sigma_1(u)}{v}+O(v^{-2}).
\end{equation}
Then the $v^{-3}$ coefficient gives
\begin{equation}
\tau_3'(u)+\frac{6a_1(u)}{r_0}\sigma_1(u)+\frac{6a_1'(u)}{r_0}=0.
\end{equation}
Using \eqref{eq:a1-tau3}, this reduces to \(a_1(u)\sigma_1(u)=0\). Hence, on
a genuine \(a_1\) branch with \(a_1\neq0\), one has
\(\sigma_1(u)=0\), so the dilaton-source function must in fact decay faster,
\(s_\phi=O(v^{-2})\).

The companion state equation \eqref{eq:state-u} then implies
\begin{equation}
r_u=\frac{a_1'(u)}{v}+O(v^{-2}),
\end{equation}
and therefore
\begin{equation}
\partial_v t_u-\frac{3a_1''(u)}{r_0\,v^2}+O(v^{-3})=0.
\end{equation}
Integrating gives
\begin{equation}
t_u=t_u^{(0)}(u)-\frac{3a_1''(u)}{r_0\,v}+O(v^{-2}),
\label{eq:a1-tu-int}
\end{equation}
where $t_u^{(0)}(u)$ is the integration function.

Finally, in the gauge where the leading $p=2$ coefficient is fixed to a constant, the $uu$ constraint removes the $O(v^0)$ piece. The \(O(v^0)\) term forces \(t_u^{(0)}(u)=0\).
At order $v^{-1}$, however, the local  FFN term $-\rho_{uu}$  contributes at the same order as the $t_u$ tail. Writing
\begin{equation}
t_u=\frac{\eta_1(u)}{v}+O(v^{-2}),
\qquad
\eta_1(u)=-\frac{3a_1''(u)}{r_0},
\label{eq:a1-eta1}
\end{equation}
the $uu$ equation gives
\begin{equation}
J_0 a_1''(u)=\frac{G}{12\pi}\bigl(\eta_1(u)-b_1''(u)\bigr).
\label{eq:a1-uu}
\end{equation}
Differentiating \eqref{eq:a1-mixed} with $A(u)=A_0$ constant gives
\begin{equation}
-\left(J_0+\frac{G}{4\pi r_0}\right)a_1''(u)-\frac{G}{12\pi}b_1''(u)=0.
\label{eq:a1-mixed-deriv}
\end{equation}
Using \eqref{eq:a1-eta1}, the relation \eqref{eq:a1-uu} becomes an identity once \eqref{eq:a1-mixed-deriv} is imposed. Thus the local asymptotic equations do not impose \(a_1''(u)=0\), so no
linearity condition \(a_1(u)=\alpha u+\beta\) follows.

We conclude that the $a_1/v$ branch is significantly more rigid than the settled $a_2/v^2$ branch in its required tail structure. It requires
\begin{equation}
t_v=\frac{\tau_3(u)}{v^3}+O(v^{-4}),
\qquad
s_\phi=O(v^{-2}),
\qquad
t_u=\frac{\eta_1(u)}{v}+O(v^{-2}),
\end{equation}
while its coefficient function remains constrained by \eqref{eq:a1-mixed} rather than being fixed to be linear. Thus the $a_1/v$ channel remains logically possible, but only in a more constrained form than the settled $a_2/v^2$ branch. This is why we regard the $a_1$ channel as more tuned and less natural than the settled $a_2$ alternative.

\bibliography{endpoint_refs}

\end{document}